\documentclass[letterpaper,12pt]{article}

\usepackage{amsthm}
\usepackage{amsfonts}
\usepackage{amsmath}
\usepackage{graphicx}
\usepackage{epstopdf}
\usepackage{lineno}
\usepackage{setspace} 
\usepackage{verbatim}
\usepackage{hyperref}

\newtheorem{theorem}{Theorem}
\newtheorem*{theorem*}{Theorem}
\newtheorem{conjecture}[theorem]{Conjecture}

\newtheorem{lemma}[theorem]{Lemma}

\newcommand{\Q}{\mathbb{Q}}

\newcommand{\ft}{\mathcal T}
\newcommand{\fc}{\mathcal C}

\newcommand{\keywords}[1]{\par\addvspace\baselineskip
\noindent{\bf Keywords:}\enspace\ignorespaces#1}

\title{Hitting all maximum cliques with a stable set using lopsided independent transversals}
\author{Andrew D.\ King\thanks{Department of Industrial Engineering and Operations Research, Columbia University, New York.  Email: $\mathtt{andrew.d.king@gmail.com}$.  Research supported by an NSERC Postdoctoral Fellowship.}}

\begin{document}

\maketitle

\begin{abstract}
Rabern recently proved that any graph with $\omega \geq \frac 34(\Delta+1)$ contains a stable set meeting all maximum cliques.  We strengthen this result, proving that such a stable set exists for any graph with $\omega > \frac 23(\Delta+1)$.  This is tight, i.e.\ the inequality in the statement must be strict.  The proof relies on finding an independent transversal in a graph partitioned into vertex sets of unequal size.
\end{abstract}


\section{Introduction and motivation}

When colouring a graph $G$, we often desire a stable set $S$ meeting every maximum clique.  For example, finding such a set $S$ efficiently is the key to colouring perfect graphs in polynomial time \cite{reed01b}.  Proving the existence of $S$ has also been very useful in attacking Reed's $\omega$, $\Delta$, and $\chi$ conjecture\footnote{$\omega$, $\Delta$, and $\chi$ denote the clique number, maximum degree, and chromatic number of a graph, respectively.}:

\begin{conjecture}[Reed \cite{reed98}]
For any graph $G$, $\chi(G)\leq \lceil \frac 12(\Delta(G)+1+\omega(G))\rceil$.
\end{conjecture}

No minimum counterexample $G$ to this conjecture has such a stable set $S$.  For if it did, it would contain a maximal stable set $S'$ meeting every maximum clique, and we would have $\lceil \frac 12(\Delta(G-S')+1+\omega(G-S'))\rceil +1 \leq \lceil \frac 12(\Delta(G)+1+\omega(G))\rceil$.  Since $S'$ is a stable set, $\chi(G)\leq \chi(G-S')+1$, contradicting the minimality of $G$.

Thus such a stable set $S$ is highly desirable when attacking Reed's conjecture for a hereditary class of graphs.  The proof of Reed's conjecture for line graphs \cite{kingrv07} exemplifies the general approach:  If the maximum degree and clique number are far apart, a combination of previously known results suffices.  If they are not far apart, we can use the structure of line graphs to prove the existence of a stable set $S$ meeting all maximum cliques.

Rabern \cite{rabern09} recently proved that if the maximum degree and clique number are close enough, we need not consider the structure of the graph class at all:

\begin{theorem}[Rabern]\label{thm:rabern}
If a graph $G$ satisfies $\omega(G)\geq \frac 34(\Delta(G)+1)$, then $G$ contains a stable set meeting all maximum cliques.
\end{theorem}

Here we prove the best possible theorem of this type:

\begin{theorem}\label{thm:main}
If a graph $G$ satisfies $\omega(G)>\frac 23(\Delta(G)+1)$, then $G$ contains a stable set $S$ meeting every maximum clique.
\end{theorem}

To see that this is best possible, let $G_k$ be the graph obtained by substituting every vertex of a 5-cycle with a clique of size $k$.  Then $\omega(G_k)=2k=\frac 23(\Delta(G_k)+1)$, and no stable set meets every maximum clique.  To prove Theorem \ref{thm:main} we apply Rabern's approach with a stronger final step.  Rabern applies Haxell's theorem \cite{haxell95}, which can be stated as follows:

\begin{theorem}[Haxell]\label{thm:haxell}
For a positive integer $k$, let $G$ be a graph with vertices partitioned into $r$ cliques of size $\geq 2k$.  If every vertex has at most $k$ neighbours outside its own clique, then $G$ contains a stable set of size $r$.
\end{theorem}

To prove our theorem we need to deal with a graph that has been partitioned into cliques of unequal size.  We use the following extension of Theorem \ref{thm:haxell}:

\begin{theorem}\label{thm:isr}
For a positive integer $k$, let $G$ be a graph with vertices partitioned into cliques $V_1,\ldots V_r$.  If for every $i$ and every $v\in V_i$, $v$ has at most $\min\{k, |V_i|-k\}$ neighbours outside $V_i$, then $G$ contains a stable set of size $r$.
\end{theorem}

Although this is not at all obvious, it is a straightforward consequence of observations made by Aharoni, Berger, and Ziv about the proof of Theorem \ref{thm:haxell} \cite{aharonibz07}.

\section{Hitting the maximum cliques}

To prove Theorem \ref{thm:main} we must investigate intersections of maximum cliques.  Given a graph $G$ and the set $\fc$ of maximum cliques in $G$, we define the {\em clique graph} $G(\fc)$ as follows.  The vertices of $G(\fc)$ are the cliques of $\fc$, and two vertices of $G(\fc)$ are adjacent if their corresponding cliques in $G$ intersect.  For a connected component $G(\fc_i)$ of $G(\fc)$, let $D_i\subseteq V(G)$ and $F_i\subseteq V(G)$  denote the union and the mutual intersection of the cliques of $\fc_i$ respectively, i.e.\ $D_i=\cup_{C\in \fc_i}C$ and $F_i=\cap_{C\in \fc_i}C$.

The proof uses three intermediate results.  The first, due to Hajnal \cite{hajnal65} (also see \cite{rabern09}), tells us that for each component of $G(\fc)$, $|D_i|+|F_i|$ is large:

\begin{lemma}[Hajnal]\label{lem:hajnal}
Let $G$ be a graph and $C_1,\ldots,C_r$ be a collection of maximum cliques in $G$.  Then
$$ \left| \bigcap_{i\leq r} C_i\right| + \left|\bigcup_{i\leq r} C_i\right| \geq 2\omega(G).$$
\end{lemma}

The second is due to Kostochka \cite{kostochka80} (proven in English in \cite{rabern09}). It tells us that if $\omega(G)$ is sufficiently close to $\Delta(G)+1$, then $|F_i|$ is large:

\begin{lemma}[Kostochka]\label{lem:kostochka}
Let $G$ be a graph with $\omega(G)>\frac 23(\Delta(G)+1)$ and let $\fc$ be the set of maximum cliques in $G$.  Then for each connected component $G(\fc_i)$ of $G(\fc)$,
$$\left| \bigcap_{C\in \fc_i} C\right| \geq 2\omega(G)-(\Delta(G)+1).$$
\end{lemma}

The third intermediate result is Theorem \ref{thm:isr}.  Combining them to prove Theorem \ref{thm:main} is a simple matter.

\begin{proof}[Proof of Theorem \ref{thm:main}]
Let $\fc$ be the set of maximum cliques of $G$, and denote the connected components of $G(\fc)$ by $G(\fc_1),\ldots,G(\fc_r)$.  For each $\fc_i$, let $F_i = \cap_{C\in \fc_i}C$ and let $D_i = \cup_{C\in \fc_i}C$.  It suffices to prove the existence of a stable set $S$ in $G$ intersecting each clique $F_i$.  

Lemma \ref{lem:kostochka} tells us that $|F_i|>\frac 13(\Delta(G)+1)$.  Consider a vertex $v\in F_i$, noting that $v$ is universal in $G[D_i]$.  By Lemma \ref{lem:hajnal}, we know that $|F_i|+|D_i|>\frac 43(\Delta(G)+1)$.  Therefore $\Delta(G)+1-|D_i|<|F_i|-\frac 13(\Delta(G)+1)$, so $v$ has fewer than $|F_i|-\frac 13(\Delta(G)+1)$ neighbours in $\cup_{j\neq i}F_i$.  Furthermore $v$ certainly has fewer than $\frac 13(\Delta(G)+1)$ neighbours in $\cup_{j\neq i}F_i$.

Now let $H$ be the subgraph of $G$ induced on $\cup_i F_i$, and let $k=\frac 13(\Delta(G)+1)$.  Clearly the cliques $F_1,\ldots,F_r$ partition $V(H)$.  A vertex $v\in F_i$ has at most $\min\{k,|F_i|-k\}$ neighbours outside $F_i$.  Therefore by Theorem \ref{thm:isr}, $H$ contains a stable set $S$ of size $r$.  This set $S$ intersects each $F_i$, and consequently it intersects every clique in $\fc$, proving the theorem.
\end{proof}

It remains to prove Theorem \ref{thm:isr}.  We do this in the next section.

\section{Independent transversals with lopsided sets}\label{sec:isr}

Suppose we are given a finite graph whose vertices are partitioned into stable sets $V_1,\ldots, V_r$.  An {\em independent system of representatives} or {\em ISR} of $(V_1,\ldots,V_r)$ is a stable set of size $r$ in $G$ intersecting each $V_i$ exactly once.  A {\em partial ISR}, then, is simply a stable set in $G$ intersecting no $V_i$ more than once.  ISR's are intimately related to both the strong chromatic number \cite{haxell04} and list colourings \cite{haxell01}.

A {\em totally dominating set} $D$ is a set of vertices such that every vertex of $G$ has a neighbour in $D$, including the vertices of $D$.  Given $J\subseteq [m]$, we use $V_J$ to denote $(V_i \mid i\in J)$.  Given $X\subseteq V(G)$, we use $I(X)$ to denote the set of partitions intersected by $X$, i.e.\ $I(X) = \{ i\in [r] \mid V_i\cap X \neq \emptyset \}$.  For an induced subgraph $H$ of $G$, we implicitly consider $H$ to inherit the partitioning of $G$.

To prove our lopsided existence condition for ISR's, we use a consequence of Haxell's proof of Theorem \ref{thm:haxell} \cite{haxell95} pointed out (and proved explicitly) by Aharoni, Berger, and Ziv \cite{aharonibz07}.  Actually we prove a slight strengthening of their result:

\begin{lemma}\label{lem:aharoni}
Let $x_1$ be a vertex in $V_r$, and suppose $G[V_{[r-1]}]$ has an ISR.  Suppose there is no $J\subseteq [r-1]$ and $D\subseteq V_J\cup \{x_1\}$ totally dominating $V_{J}\cup \{x_1\}$ with the following properties:
\begin{enumerate}
\item $D$ is the union of disjoint stable sets $X$ and $Y$.
\item $Y$ is a (not necessarily proper) partial ISR for $V_J$.  Thus $|Y|\leq |J|$.
\item Every vertex in $Y$ has exactly one neighbour in $X$.  Thus $|X|\leq |Y|$.
\item $X$ contains $x_1$.
\end{enumerate}
Then $G$ has an ISR containing $\{x_1\}$.
\end{lemma}

\begin{proof}
Let $G$ be a minimum counterexample; we can assume $G = G[V_{[r-1]}\cup \{x_1\}]$.  Furthermore, $r>1$ otherwise the lemma is trivial.  Let $R_1$ be an ISR of $G[V_{[r-1]}]$ chosen such that the set $Y'_1 = Y_1 = R_1 \cap N(x_1)$ has minimum size.  We know that $R_1$ exists because $G[V_{[r-1]}]$ has at least one ISR, and we know that $Y'_1$ is nonempty because $G$ does not have an ISR.  Now let $X_1 = \{x_1\}$ and let $D_1=X_1\cup Y_1$.

We now construct an infinite sequence of partial ISRs $Y_1 \subset Y_2 \subset \ldots$, which contradicts the fact that $G$ is finite.  Let $i > 1$, and suppose we have sets $\{R_j, Y_j, X_j \mid 1\leq j <i\}$ such that:
\begin{itemize}
\item $X_j$ is a stable set consisting of distinct vertices $\{x_1, \ldots, x_j\}$.  For $j>1$, $x_j$ is a vertex in $G[V_{I(Y_{j-1})}]$ with no neighbour in $X_{j-1}\cup Y_{j-1}$.

\item $R_j$ is an ISR of $G[V_{[r-1]}]$ such that for every $1 \leq \ell < j$, $R_j \cap N(X_\ell) = Y_\ell$.  Subject to that, $R_j$ is chosen so that $Y'_j = R_j \cap N(x_j)$ is minimum.  For $1 \leq j < i$, $Y'_j$ is nonempty.

\item $Y_j = \cup_{i=1}^j Y'_j$.
\end{itemize}

To find $x_i$, $Y'_i$, and $R_i$, we proceed as follows.
\begin{enumerate}
\item Let $x_i$ be any vertex in $G[V_{I(Y_{i-1})}]$ with no neighbour in $X_{i-1}\cup Y_{i-1}$.  We know that $x_i$ exists, otherwise the set $D_{i-1}=X_{i-1}\cup Y_{i-1}$ would be a total dominating set for $G[V_{I(Y_{i-1})} \cup \{x_1\}]$, contradicting the fact that $G$ is a counterexample.

\item Let $R_i$ be an ISR of $G[V_{[r-1]}]$ chosen so that for all $1\leq j < i$, $R_i \cap N(x_j) = R_j \cap N(x_j) = Y'_j$.  Subject to that, choose $R_i$ so that $Y'_i = R_i\cap N(x_i)$ is minimum.  We know that $R_i$ exists because $R_{i-1}$ is a possible candidate for the ISR.

\item It remains to show that $Y'_i$ is nonempty, i.e.\ that $Y_i \neq Y_{i-1}$.  Suppose $Y'_i = \emptyset$.  We will show that this contradicts our choice of $R_j$ for the unique $j<i$ such that $x_i \in V_{I(Y'_j)}$.  Let $y$ be the unique vertex in $R_i \cap V_{I(x_i)}$.  Construct $R'_j$ from $R_i$ by removing $y$ and inserting $x_i$.  Now for every $\ell$ such that $1\leq \ell < j$, $R'_j \cap N(x_\ell) = Y'_\ell = R_j \cap N(x_\ell)$.  For $j$, $R'_j \cap N(x_j) = (R_j \cap N(x_j))\setminus \{ y\}$, a contradiction.  Thus $Y'_i$ is nonempty.

\item Set $X_i = X_{i-1}\cup \{x_i\}$ and $Y_i = Y_{i-1}\cup Y'_i$.
\end{enumerate}
This choice of $X_i$, $R_i$, and $Y_i$ sets up the conditions so that we can repeat our argument indefinitely for increasing $i$, a contradiction since $G$ is finite.
\end{proof}

The lemma easily implies Theorem 3.5 in \cite{aharonibz07}, and allows us to prove a strengthening of Theorem \ref{thm:isr}:

\begin{theorem}\label{thm:isr2}
Let $k$ be a positive integer and let $G$ be a graph partitioned into stable sets $(V_1,\ldots,V_r)$.  If for each $i\in [r]$, each vertex in $V_i$ has degree at most $\min\{k,|V_i|-k\}$, then for any vertex $v$, $G$ has an ISR containing $v$.
\end{theorem}

\begin{proof}
Suppose $G$ is a minimum counterexample for a given value of $k$.  Clearly we can assume each $V_i$ has size greater than $k$, and that $G[V_J]$ has an ISR for all $J\subset [r]$.  Take $v$ such that $G$ does not have an ISR containing $v$; we can assume $v\in V_r$.  By Lemma \ref{lem:aharoni}, there is some $J \subseteq [r-1]$ and a set $D\subseteq V_J\cup \{v\}$ totally dominating $V_J\cup \{v\}$ such that (i) $D$ is the union of disjoint stable sets $X$ and $Y$, (ii) $Y$ is a partial ISR of $V_J$, (iii) $|X|\leq |Y|\leq |J|$, and (iv) $v\in X$.

Since $D$ totally dominates $V_J\cup \{v\}$, the sum of degrees of vertices in $D$ must be greater than the number of vertices in $V_J$.  That is, $\sum_{v\in D}d(v) > \sum_{i\in J}|V_i|$.  Clearly $\sum_{v\in X}d(v) \leq k\cdot |J|$ and $\sum_{v\in Y}d(v) \leq \sum_{i\in J}(|V_i|-k)$, so $\sum_{v\in D}d(v) \leq \sum_{i\in J}|V_i|$, contradicting the fact that $D$ is totally dominating.  This proves the theorem.
\end{proof}

This extends Haxell's theorem by bounding the difference between the degree of a vertex and the size of its partition.  One might hope that bounding the ratio of these by $\frac 12$ is enough, but it is not:  Given $V_1$ of size four and $V_2,\ldots,V_5$ of size two, in which each vertex of $V_1$ dominates one of the smaller sets, there exists no ISR \cite{haxellpc}.  Lemma \ref{lem:aharoni} cannot imply such a result because in the totally dominating set $D = X\cup Y$, we have no control over the average degree of a vertex in $X$ -- it may be $k$.  So while we know that the average degree of a vertex in $Y$ behaves nicely with respect to the average partition size, the same is not necessarily true of $X$.  Thus Theorem \ref{thm:isr2} gives a lopsided existence condition that is not only a useful consequence of Lemma \ref{lem:aharoni}, but also a natural one.

\section{Acknowledgements}

The author is grateful to Landon Rabern and Penny Haxell for helpful discussions, and to Robert Himmelmann for pointing out an error in an earlier version of the paper.

\bibliography{masterbib}
\end{document}